\documentclass[article]{amsart}
\usepackage{amsmath,amsthm}
\usepackage{amssymb}
\usepackage{bbm}
\usepackage{tikz}
\usetikzlibrary{arrows,decorations.markings,chains,calc}
\newtheorem{corollary}{Corollary}
\newtheorem{proposition}{Proposition}
\newtheorem{lemma}{Lemma}

\newtheorem{theorem}{Theorem}

\newtheorem{remark}{Remark}

\newcommand{\N}{\mathbb{N}} % Natural numbers
\newcommand{\C}{\mathbb{C}} % Complex numbers
\newcommand{\R}{\mathbb{R}} % Real numbers
\newcommand{\nn}{\|}
\newcommand{\bx}{{\bf x}}

\newcommand{\SP}[2]{\big\langle #1,#2 \big\rangle} %S. Product BIG 
\newcommand{\sps}[2]{\langle #1,#2 \rangle} %Scalar Product
\newcommand{\ri}{i}

\newcommand{\no}{{\boldsymbol n}}
\newcommand{\fh}{\hat{f}}
\newcommand{\pauli}{\boldsymbol{\sigma}}
\newcommand{\vg}{\boldsymbol{\gamma}}

\newcommand{\ve}{\varepsilon}
\newcommand{\M}{M}

\renewcommand{\le}{\leqslant}

 \renewcommand{\ge}{\geqslant}
      
\title{Infinite mass boundary conditions for Dirac operators}
\author{Edgardo Stockmeyer}
\address{Edgardo Stockmeyer\\
Institute of Physics\\
Pontificia Universidad Cat\'olica de Chile\\
Vicu\~na Mackenna 4860\\
7820436 Santiago, Chile.}
\email{stock@fis.puc.cl}
\author{Semjon Vugalter}
\address{Semjon Vugalter\\
Karlsruhe Institute of Technology\\ 
Englerstrasse 2\\
 76131 Karlsruhe, Germany.}
\email{semjon.wugalter@kit.edu}
\begin{document}
\maketitle
\begin{abstract}
  We study a  self-adjoint realization of a massless Dirac
  operator on a bounded connected domain $\Omega\subset \R^2$ which is
  frequently used to model graphene. In
  particular, we show that this operator is the limit, as $M\to \infty$, of a
  Dirac operator defined on the whole plane, with a mass term of size
  $M$ supported outside $\Omega$. 
\end{abstract}
\section{Introduction}
% It is by now well-known that the low energy excitations corresponding to
% charge-carriers in graphene are modelled by a two-dimensional
% massless Dirac operator with the speed of light $c$ replaced by the
% Fermi velocity, $v_F\sim 10^{-2} c$
% \cite{wallace,novoselov2004electric,Fefferman,castro2009electronic}. One
% topic of intense research activity in physics is quantum dots or
% particle confinement in graphene. Interestingly Dirac particles can
% not be trivially confined through electric potential walls as it is
% the case for usual non- relativistic particles. Nevertheless,
% confinement can be achieved using for instance non-homogeneous
% magnetic fields \cite{de2007magnetic} or by simply cutting a graphene
% sample. The latter is modelled through one (or two copies of a) Dirac
% operator with appropriate boundary conditions
% \cite{akhmerov2008boundary}.

Consider a bounded domain $\Omega\subset \R^2$. It is known that a
Dirac operator $H$ can not be self-adjointly realized in
$L^2(\Omega,\R^2)$ by imposing Dirichlet boundary
conditions.  In 1987, Berry and Mondragon initiated  the study of
 self-adjoint realizations of  Dirac operators under the condition
 that the  normal projection of the current density vanishes at the boundary
$\partial\Omega$ \cite{Berry1987}. This condition can be
mathematically stated as 
\begin{align}
  \label{eq:36}
  \no(\bx)\cdot\big(\varphi(\bx),
  \pauli \varphi(\bx)\big)_{\C^2}=0,\qquad \bx \in \partial\Omega,
\end{align}
where $\no\in\R^2$ is the outward normal vector  to $\partial\Omega$,  
$\varphi\in L^2(\Omega,\R^2)$, and
$\pauli=(\sigma_1,\sigma_2)$ is a vector formed by the usual
Pauli matrices 
\begin{align*}
\sigma_1
=\left(
\begin{array}{cc}
 0&1\\
 1&0
\end{array}
 \right),\quad
\sigma_2=\left(
\begin{array}{cc}
0&-i\\
i&0
\end{array}
\right).
\end{align*}
Equation \eqref{eq:36} gives rise to a whole family of different
boundary conditions (see Equation \eqref{eq:34} below). In this
present work we focus on one of these self-adjoint realizations,
denoted by $H_\infty$, which corresponds to the so-called infinite
mass boundary conditions. In the physics literature, the operator
$H_\infty$ has gained renewed interest due to its application to model
quantum dots  in graphene 
\cite{castro2009electronic,phys1,phys2,phys3,phys4,phys5,akhmerov2008boundary}.

Let $H_M$ be the Dirac operator defined on $\R^2$ with a mass 
$M$ on $\R^2\setminus\Omega$, and $0$ inside $\Omega$.  In
\cite{Berry1987} it was shown that certain plane-wave solutions of the
eigenvalue equation $H_M\psi=E\psi$, in the limit $M\to\infty$,
satisfy the same boundary conditions as the eigenfunctions of
$H_\infty$.  The main result of this work, Theorem \ref{mainthm}, is
the convergence, in the sense of spectral projections, of $H_M$ towards
$H_\infty$.
\subsection{Definitions and main result}\label{ppp}
Let us introduce some notation used throughout this article. We denote
by $\Omega\subset \R^2$ a bounded connected domain with boundary
$\partial\Omega\in C^3$ of length $L>0$. We parametrize $\partial
\Omega$ by the curve $\vg: [0,L]\to \partial\Omega $ in its
arc-length, i.e., $|\vg'(s)|=1$. For a given self-adjoint operator $H$,
we denote by $\sigma(H)$ its spectrum, and by $E_I(H)$ its spectral
projection on the set $I\subset \R$.  We use the symbols
$\sps{\cdot}{\cdot}$ and $(\cdot,\cdot)$ to denote the scalar products
in $L^2$ and $\C^2$, respectively. Moreover, we use $\|\cdot\|,\,\,
\|\cdot\|_\Omega$ and $\|\cdot\|_{\partial \Omega}$ for the
$L^2$-norms in $\R^2$, $\Omega$, and ${\partial \Omega}$,
respectively. We drop the indication to the domain of integration if
it is clear from the context.  In particular,
\begin{align*}
  \nn\varphi\nn^2_{\partial\Omega}=\int_{\partial\Omega}  |\varphi(\bx)|^2\,
  d\omega(\bx)=\int_0^L |\varphi(\vg(s))|^2 ds.
\end{align*}
%%%%
%%%%
Let  $T$ be  the differential expression associated with the
massless Dirac operator, i.e., 
\begin{align*}
  T=\frac{1}{i}\pauli\cdot\nabla=\frac{1}{i}(\partial_1\sigma_1+\partial_2\sigma_2)=\frac{1}{i}
\left(
\begin{array}{cc}
0&\partial_1-i\partial_2\\
\partial_1+i\partial_2&0
\end{array}
\right).
\end{align*}
%%%%
%%%%
It is interesting to identify the boundary conditions needed to
realize $T$ as a self-adjoint operator in $L^2(\Omega,\C^2)$: For
$\varphi\in C^\infty(\overline{\Omega},\C^2)$, we compute
\begin{align*}
  \sps{\varphi}{T\varphi}&= \sps{\varphi}{\tfrac{1}{i} \pauli\cdot\nabla
    \varphi}= \sps{\tfrac{1}{i} \pauli\cdot\nabla \varphi}{\varphi}
-i \int_\Omega \nabla\cdot(\pauli \varphi (\bx),
\varphi(\bx)) d\bx\\
&= \sps{T \varphi}{\varphi}
 -i\int_{\partial\Omega}  {\boldsymbol J}_\varphi(\bx)\cdot {\boldsymbol n} \,d\omega(\bx),
\end{align*}
where, in the last equality, we use Green's formula. Here, $ {\boldsymbol
  J}_\varphi(\bx):= (\varphi (\bx), \pauli \varphi(\bx)) $, and
${\boldsymbol n}$ is the outward normal vector of $\Omega$. Hence, any
self-adjoint realization of $T$ must satisfy
\begin{align*}
  \int_{\partial\Omega}  {\boldsymbol J}_\varphi\cdot {\boldsymbol n}\,
  d\omega(\bx)=0.
\end{align*}
Note that the commutator $[T,x_j]=\sigma_j$. Thus, in view of
Heisenberg's evolution equation, we may interpret $ {\boldsymbol
  J}_\varphi(\bx)$ as the current density. As noted in \cite{Berry1987}, it
is straightforward to see that $ {\boldsymbol
  J}_\varphi(\bx)$ vanishes pointwise if and only if the components of
$\varphi$ satisfy
\begin{align}
  \label{eq:34}
  \varphi_2(\vg(s))=i B(s) e^{i\alpha(s)}
  \varphi_1(\vg(s)),\qquad s\in [0,L),
\end{align}
for some real function $B$, or when $ \varphi_1$ equals zero at the
boundary. Here $\alpha(s)$ is the  turning angle, i.e.,
the angle between ${\boldsymbol n}$ and the
$x_1$-axis at the point $\vg(s)\in \partial \Omega$.
%%%%

%%%%
In this article we focus on the case $B=1$. In order to define the
operator, let us first write the corresponding condition \eqref{eq:34}
in a more compact form that will become useful later on. For
$s\in[0,L)$ define
$a(s):=i e^{i\alpha(s)}$ and consider the matrix
\begin{equation}
  \label{eq:8}
  A(s):=\left(\begin{array}{cc}
0 & a(s)^*\\
a(s)&0
\end{array}\right).
\end{equation}
Clearly, $A(s)$ has eigenvalues $1$ and $-1$. We define the
corresponding eigenprojections as
\begin{equation}\label{projections}
P_\pm(s)=(1\pm A(s))/2.
\end{equation}
It is easy to see that condition  \eqref{eq:34}, for $B=1$, is
equivalent to $P_-(s)\varphi(\vg(s))=0$. 
Let
\begin{align*}
  \mathcal{D}_\infty :=\{\varphi\in
  H^1(\Omega,\C^2) : P_-(s)\varphi(\vg(s))=0, s\in [0,L) \}.
\end{align*}
We define the operator 
\begin{align*}
&H_\infty:\mathcal{D}_\infty \subset
L^2(\Omega,\C^2)\to L^2(\Omega,\C^2),\\ 
&\qquad \qquad H_\infty \varphi=T \varphi.
\end{align*}
It is known that $H_\infty$ is self-adjoint and that its spectrum is
purely discrete (see Proposition \ref{thm0}, Remark \ref{rmkthm0}, and
Proposition \ref{thm1},  from Section \ref{secinf}, for further details).
%%%%%%% 

%%%%%%%
In order to state the main result of the work at hand, Theorem
\ref{mainthm} below, we introduce the Dirac operator defined on $\R^2$
with a mass term supported outside $\Omega$. For $M>0$, we define 
\begin{align*}
  & H_M: H^1(\R^2,\C^2) \subset L^2(\R^2,\C^2) \to L^2(\R^2,\C^2), \\
  &\qquad\quad H_M\psi=T\psi+\sigma_3M(1-\mathbbm{1}_\Omega)\psi,
\end{align*}
where $\mathbbm{1}_\Omega$ is the characteristic function on $\Omega$
and $\sigma_3=i\sigma_2 \sigma_1$. It is easy to see that $H_M$ is
self-adjoint and has purely discrete spectrum on the interval $(-M,M)$
(see Lemma \ref{thm2} below).
%%%%%%%%%%%
%%%%%%%%%%%

%%%%%%%%%
We are now in position to state the main result of our work.
%%%%%%%%%%%
\begin{theorem}[Convergence of Spectral Projections]\label{mainthm}
  Let $\Omega$ be a connected bounded domain with a $C^3$-boundary.
  Let $\lambda\in \R$ be an eigenvalue of $H_\infty$. Then, for any
  $0<\varepsilon< {\rm dist}(\lambda,\sigma(H_\infty)\setminus
  \{\lambda\})$, we have
\begin{align}
  \label{eq:29.3}
  \big\|  \widetilde{E}_{\{\lambda\}}(H_\infty)-E_{(\lambda-\varepsilon,
    \lambda+\varepsilon)}(H_M)  \big\|\to 0\quad \mbox{as}\quad M\to \infty,
\end{align}
where $ \widetilde{E}_{\{\lambda\}} (H_\infty)={E}_{\{\lambda\}}
(H_\infty)\oplus \{0\}$ with respect to the splitting
$\mathcal{H}=L^2(\Omega,\C^2)\oplus L^2(\R^2\setminus \Omega,\C^2)$.
In particular, as $M\to \infty$, the eigenvalues of $H_M$ converge
towards the eigenvalues of $H_\infty$ and any eigenvalue of $H_\infty$
is the limit of eigenvalues of $H_M$.
\end{theorem}
\begin{remark}
(i) The required $C^3$-regularity of the boundary is due to the
 application of our regularity result Theorem \ref{regularity} below. \\
(ii) One can easily see that $H_{-M}$ converges, as $M\to
   \infty$, to the Dirac operator with the boundary condition \eqref{eq:34} 
  with $B=-1$. This can be shown using
  the antiunitary transformation $U=i\sigma_2 \mathcal{C}$. Indeed,
  $UH_M U^{-1}=H_{-M}$ and if $\varphi\in \mathcal{D}_\infty$ then $U\varphi=\widetilde{\varphi}$ with
$\widetilde{\varphi}_2(\vg(s))=-i e^{i\alpha(s)}
\widetilde{\varphi}_1(\vg(s))$ holds.
\end{remark}
Let us briefly describe the strategy of the proof of the main
result. We start by observing that both operators $H_\infty$ and $H_M$
have symmetric spectra with respect to zero (see Proposition
\ref{thm1} and Lemma \ref{thm2}). This enables us to study,
instead, the spectra of the positive operators $H_\infty^2$ and $H_M^2$
and to apply the minimax principle.

Next, we give a lower bound for the quadratic form
$\sps{H_M\psi}{H_M\psi}$, which allows us to show that a function
$\psi\in E_{(-A,A)}(H_M)L^2(\R^2,\C^2)$, for fixed $A>0$ and
$M\to\infty$, should satisfy  $\|\psi\|_{\R^2\setminus \Omega}\to
0$ and that $\|P_- \psi\|_{\partial\Omega} \to 0$. In other words, in
the limit $M\to\infty$, the function $\psi$ is supported inside
$\Omega$ and satisfies the infinite mass boundary conditions (see
lemmas \ref{qfm} and \ref{coro1.1} and Corollary \ref{coro1}).

The next step goes as follows: Given $A\notin\sigma(H_\infty)$ and
$\varphi $ from the range of $E_{(-A,A)}(H_\infty)$, we construct a
trial function $\psi_M\in\mathcal{D}(H_M)$ with
$\psi_M\!\!\upharpoonright_\Omega =\varphi$, $\psi_M$ exponentially
small outside $\Omega$, and having the property
$\sps{H_M\psi}{H_M\psi}<A^2+\epsilon(M)$ with $\epsilon(M)$ tending to
zero as $M\to\infty$ (see Lemma \ref{L1}). For sufficiently large
$M>1$, this implies that the dimension of the range of
$E_{(-A,A)}(H_M)$ is at least as large as that of
$E_{(-A,A)}(H_\infty)$. This construction uses some regularity
properties of eigenfunction of the operator $H_\infty$ presented in
Theorem \ref{regularity}.

To get the converse statement,
we construct in Lemma \ref{comparation} a
function $\varphi\in\mathcal{D}_\infty$ from a given $\psi\in
E_{(-A,A)}(H_M)L^2(\R^2,\C^2)$, with $\sps{H_\infty\psi}{H_\infty\psi}<A^2+\epsilon(M)$, such that
$\|\varphi-\psi_M\|_\Omega$ and $\epsilon(M)$ tend to zero as
$M\to\infty$. Finally, Lemma \ref{lastlemma} completes the proof of the theorem.

\section{Properties of $H_\infty$}\label{secinf}
We start by stating some general facts on $H_\infty$, namely
its self-adjointness and the discreteness of its spectrum. 
\begin{proposition}\label{thm0}
  Let $\Omega$ be a  domain with a
  $C^2$-boundary. Then the operator $H_\infty$ defined above is
self-adjoint on $\mathcal{D}_\infty$.
\end{proposition}  
\begin{remark}\label{rmkthm0}
  A similar statement, for a more general class of Dirac operators in
  domains with $C^\infty$-boundaries, can be found in \cite[Lemma
  1]{Prokhorova2013}. Note however that the most difficult part of the
  proof, namely to show that the domain of the adjoint operator is
  contained in $H^1(\Omega,\C^2)$, can be found in \cite{BLZ2009} and
  in the references therein. A more direct proof, which holds for
  $C^2$-boundaries, is given in \cite{BFSV2016}.
\end{remark}
%%%%%%%
%%%%%%%
Due to the compact embedding of $H^1(\Omega)$ in $L^2(\Omega)$ we have
that the spectrum of $H_\infty$ is discrete.  Moreover, it is
straightforward to see that $\sigma(H_\infty)$ is symmetric with
respect to zero. Indeed, define $U:=\sigma_1\mathcal{C}$ where
$\mathcal{C}$ is the complex conjugation on $L^2(\Omega,\C^2)$. It is
clear that $U$ is antiunitary and leaves $\mathcal{D}_\infty$
invariant. That the spectrum is symmetric now follows from the
relation $UH_\infty\varphi=-H_\infty U\varphi,
\varphi\in\mathcal{D}_\infty$. We summarize these observations in the
following statement.
\begin{proposition}\label{thm1}
The operator $H_\infty$ has purely discrete spectrum and its spectrum
is symmetric with respect to zero, that is, 
\begin{center} 
$E\in \sigma(H_\infty)$ if and only if $-E\in
   \sigma(H_\infty)$.
\end{center}
\end{proposition}
%%%%%%%
%%%%%%%
The proof of the next result can be found in Appendix
\ref{appendixa1}.
\begin{lemma}\label{energyi}
For any $\varphi\in \mathcal{D}_\infty$ we have 
\begin{align}\label{qf}
  \|H_\infty\varphi\|^2=\int_\Omega |\nabla\varphi(\bx)|^2
  d\bx+\frac{1}{2} \int_0^L \alpha'(s)
  |\varphi(\vg(s))|^2 ds.
\end{align}
\end{lemma}
\subsection{Regularity of eigenfunctions}
For a fix $0<\delta<1/\|\alpha'\|_\infty$ we define
a neighbourhood $Q_0$ of $\partial\Omega$ as
\begin{align}\label{q}
  Q_0:=\{ x\in \R^2 : {\rm dist}(x,\partial\Omega)<\delta\}.
\end{align}
In this set we can use the direction of normal and tangent vectors to
$\partial\Omega$ as local system of coordinates $(r,s)$. Indeed, the  coordinates map is given by
\begin{equation}
  \label{eqq:2.1}
\begin{split}
  &{\boldsymbol \kappa} :   (-\delta,\delta)\times [0,L) \to \R^2\\
&\quad {\boldsymbol \kappa}(r,s)=\vg(s)+r\no(s).
\end{split}
\end{equation}
Using that $$ \no(s)=(\cos\alpha(s),\sin\alpha(s))\,\,\,\mbox{and}\,\,\,\vg'(s)=(-\sin\alpha(s), \cos\alpha(s)),$$ we readly
obtain that $ \no'(s)=\alpha'(s)\vg'(s)$ and
\begin{align*}
&  \partial_r {\boldsymbol \kappa}(r,s)=\no(s)  \\
&  \partial_s {\boldsymbol \kappa}(r,s)=\vg'(s)(1+r\alpha'(s)).
\end{align*}
The Jacobian of the coordinates map is $(1+\alpha'r)$. Thus
${\boldsymbol \kappa}$ is a $C^1$-diffeomorphismus whenever
$\delta<1/\|\alpha'\|_\infty$. Let us now relate derivatives in
different coordinates. We have
\begin{equation}
  \label{eqq:3}
  \begin{split}
\partial_r+i \partial_s&=(\tfrac{\partial {\kappa}_1}{\partial r}    +i
\tfrac{\partial {\kappa}_1}{\partial s} ) \partial_1+
(\tfrac{\partial { \kappa}_2}{\partial r}    +i
\tfrac{\partial { \kappa}_2}{\partial s} ) \partial_2\\
&=e^{-i\alpha(s)}(\partial_1+i \partial_2)+ir\alpha'(s)(\cos\alpha(s)\partial_2-
\sin\alpha(s)\partial_1).
\end{split}
\end{equation}
Analogously we obtain that
\begin{align}
  \label{eq:4}
  \partial_r-i \partial_s=e^{i\alpha(s)}(\partial_1-i \partial_2)-ir\alpha'(s)(\cos\alpha(s)\partial_2-
\sin\alpha(s)\partial_1).
\end{align}
This can be further simplified using the  identity
$
i(\cos\alpha(s)\partial_2-
\sin\alpha(s)\partial_1)=\frac{1}{2}\big[ e^{-i\alpha(s)}(\partial_1+i \partial_2)- e^{i\alpha(s)}(\partial_1-i \partial_2)\big]
$. We obtain that
\begin{align}\label{grads}
\left( 
 \begin{array}{cc}
\partial_r\\
(1+r\alpha')\partial_s
  \end{array}
\right)=
\left( 
 \begin{array}{cc}
\cos\alpha&\sin\alpha\\
-\sin\alpha&\cos\alpha
  \end{array}
\right)
\left( 
 \begin{array}{cc}
\partial_1\\
\partial_2
  \end{array}
\right).
\end{align}
Our next result is on the regularity of solutions $\varphi\in
\mathcal{D}_\infty$ of the following eigenvalue problem
\begin{equation}
  \label{eq:28}
  \begin{split}
&  H_\infty\varphi=E\varphi \quad \mbox{in} \quad L^2(\Omega, \C^2)\\
& P_-(s)\varphi(\vg(s))=0, \quad\mbox{for almost all} \quad s\in[0,L].
\end{split}
\end{equation}
\begin{theorem}\label{regularity}
Let $\Omega$ be a domain with  $C^3$-boundary. If $\varphi\in
\mathcal{D}_\infty$ is a solution of the eigenvalue problem
\eqref{eq:28}, then $\varphi\in H^2(\Omega,\C^2)$.
\end{theorem}
\begin{proof} 
We define the following operator in $L^2(Q_0,\C^2)$
\begin{equation}
  \label{eq:6}
   (P\varphi)(\bx)=
\tfrac{1}{2}(1-A(s))\varphi({\boldsymbol \kappa}(r,s)), \quad
\bx={\boldsymbol \kappa}(r,s)\in Q_0\cap \Omega,
\end{equation}
where $A$ is the matrix function defined in \eqref{eq:8}.
Then, for $\varphi$ satisfying \eqref{eq:28}, we have
\begin{equation}
\label{ee}
\begin{split}
&  H_\infty\varphi=E\varphi \quad \mbox{in} \quad L^2(\Omega,\C^2)\\
& P\varphi=0 \quad \mbox{on} \quad \partial\Omega.
\end{split}
\end{equation}
Since $-\Delta\varphi=E^2 \varphi$ holds, in a distributional sense,
we have using \cite[Theorem 8.8]{GT1983} that $\varphi$ is in $ H^2$ on
the interior of $\Omega$.  For $x_0\in \partial\Omega$ we denote by
$B_\rho(x_0)$ the open ball around $x_0$ with radius $\rho>0$. Let
$\chi \in C^\infty(\R^2, [0,1])$ supported on $B_{2\rho}(x_0))$ with
$\chi=0$ on $\R^2\setminus B_{2\rho}(x_0)$ and $\chi=1$ on
$B_\rho(x_0) $. We choose $\rho<\delta/2$.

Next we show, using the eigenvalue equation \eqref{ee},  that
  \begin{equation}\label{aaaa1}
\begin{split}
&  -\Delta( P \chi\varphi)=f \quad \mbox{in} \quad \Omega\\
& P \chi\varphi=0 \quad \mbox{on} \quad \partial\Omega,
\end{split}
\end{equation}
holds  for some $f\in L^2(\Omega,\C^2)$.
To this end we note that
\begin{equation}
\label{riga}
\begin{split}
 H_\infty( P \chi\varphi)&=EP\chi\varphi+[H_\infty,P\chi]\varphi\\
&=\big( EP\chi+\tfrac{1}{2} [H_\infty, \chi]\big)\varphi-\tfrac{1}{2}[H_\infty,A\chi]\varphi,
\end{split}
\end{equation}
where $[\cdot,\cdot]$  denotes the commutator. 
Since $ EP\chi+\tfrac{1}{2} [H_\infty, \chi]$ is continuously
differentiable and $\varphi\in H^1(\Omega,\C^2)$ we see that the first
term above is in $H^1(\Omega,\C^2)$. A direct computation
shows 
\begin{align*}
  [H_\infty,A\chi]=\left(\begin{array}{cc}
d^*\beta-\beta^* d&0\\0&d\beta^*-\beta d^*
    \end{array}
\right),
\end{align*}
where $d:=-i(\partial_1+i\partial_2)$ and $\beta:=a\chi$ (see \eqref{eq:8}).
We have that
\begin{align}\label{gggttt}
H_\infty [H_\infty,A\chi]\varphi=\big\{H_\infty, [H_\infty,A\chi]\big\}\varphi-E[H_\infty,A\chi]\varphi,
\end{align}
where $\{\cdot, \cdot\}$ denotes the anticommutator. Observe that the
second term on the right hand side of \eqref{gggttt} is obviously
square integrable. Moreover, for the first term we find that
\begin{align}
  \label{eq:48}
  \big\{H_\infty, [H_\infty,A\chi]\big\}
=
\left(
\begin{array}{cc}
0 & [-\Delta,\beta^* ]
\\
\,[-\Delta,\beta]&0 
  \end{array}
\right).
\end{align}
Since  $\alpha\in C^2$ we get using \eqref{gggttt} that $H_\infty [H_\infty,A\chi]\varphi
\in L^2(\Omega,\C^2)$. Applying $H_\infty$ to the l.h.s. and the
r.h.s. of \eqref{riga}  yields  \eqref{aaaa1} for
\begin{align*}
  f=H_\infty
\big( EP\chi+\tfrac{1}{2} [H_\infty, \chi]\big)\varphi-
\tfrac{1}{2} H_\infty [H_\infty,A\chi]\varphi \in L^2(\Omega, \C^2).
\end{align*}
Equation  \eqref{aaaa1} implies by  \cite[Theorem 8.12]{GT1983}  that
$P \chi\varphi\in
H^2_0(\Omega,\C^2)$. As a consequence we get 
% $  \|\chi(\varphi_2-ie^{i\alpha} \varphi_1)\|_{H^2(\Omega)}\le
%     C\|\varphi\|_{L^2(\Omega)}.$
% and we get the bound
\begin{align}\label{eq2.1}
  \chi(\varphi_2-ie^{i\alpha} \varphi_1)\in
H^2_0(\Omega).
\end{align}
In particular, $\varphi_2-ie^{i\alpha} \varphi_1\in
H^2(B_\rho(x_0)\cap \Omega)$. Since the boundary can be covered by
finitely many balls and interior regularity holds we get, writing
$U:=Q_0\cap \Omega$,
 \begin{align}\label{eq2}
  \varphi_2-ie^{i\alpha} \varphi_1\in
H^2(U).
\end{align}
According to \cite{BFSV2016} we have that $E\not=0$.  Substituting the
eigenvalue equation  $\varphi=E^{-1}
H\varphi$ in \eqref{eq2} we get
\begin{equation*}
  (\partial_1+i\partial_2)\varphi_1-ie^{i\alpha}(\partial_1-i\partial_2)\varphi_2
  \in H^2(U).
\end{equation*}
It follows from this  that
\begin{align}
  \label{eq:1.1}
  e^{-i\alpha}(\partial_1+i\partial_2)\varphi_1-i(\partial_1-i\partial_2)\varphi_2 \in H^2(U).
\end{align}
Since $(\partial_1-i\partial_2)(\varphi_2-ie^{i\alpha}
\varphi_1)\in H^1(U)$ by \eqref{eq2}, we get that 
\begin{align*}
 &  e^{-i\alpha}(\partial_1+i\partial_2)\varphi_1-i(\partial_1-i\partial_2)
\varphi_2
+i(\partial_1-i\partial_2)(\varphi_2-ie^{i\alpha}
\varphi_1)\\
&\qquad= e^{-i\alpha} (\partial_1+i\partial_2)\varphi_1
+(\partial_1-i\partial_2)(e^{i\alpha}\varphi_1)
\end{align*}
belongs to $H^1(U)$. Since $\alpha\in C^2$ we find that 
\begin{align*}
  e^{-i\alpha} (\partial_1+i\partial_2)\varphi_1
+e^{i\alpha} (\partial_1-i\partial_2)\varphi_1 \in H^1(U).
\end{align*}
Finally in view of equations \eqref{eqq:3} and \eqref{eq:4} we see
that the latter expression equals $2\partial_r \varphi_1$. This
implies that $\partial_r^2 \varphi_1$ and $\partial_s\partial_{r} \varphi_1$
are square integrable in $U$. Since $-\Delta\varphi_1=E^2\varphi_1$
holds  we obtain also that $\partial_s^2
\varphi_1 \in L^2(U)$. That $\varphi_1\in H^2(\Omega)$ follows from
this and interior regularity. The analog statement for $\varphi_2$ can
be deduced from the latter together with \eqref{eq2}. This completes
the proof.
\end{proof}
\section{Properties of $H_M$}
We start by computing the quadratic energy of $H_M$.
\begin{lemma}\label{energym}
For and $\psi \in H^1(\R^2,\C^2)$ we have 
\begin{equation}
  \label{eq:11}
\begin{split}
  \|H_M\psi\|^2&=\int_{\R^2} |\nabla \psi (\bx)|^2d\bx+
M^2 \int_{\R^2\setminus\Omega} |\psi (\bx)|^2d\bx\\
&\quad-M\int_0^L \big[|P_+\psi(\vg(s))|^2-
|P_-\psi(\vg(s))|^2\big]  ds,
\end{split}
\end{equation}
where $P_\pm$ are the projections defined in \eqref{projections}.
\end{lemma}
\begin{proof}
For  $\psi = (\psi_1,\psi_2)^{\rm T}\in H^1(\R^2,\C^2)$,  a direct computation shows that
\begin{align*}
  \|H_M\psi\|^2&=\int_\Omega |\nabla \psi (\bx)|^2d\bx+
M^2 \int_{\R^2\setminus\Omega} |\psi (\bx)|^2d\bx+
2M{\rm Re}\SP{\tfrac{1}{i}\pauli\cdot\nabla \psi}{\sigma_3(1-\mathbbm{1}_\Omega)\psi}.
\end{align*}
Applying Green's identity we find that
\begin{align*}
  2M{\rm Re}\SP{\tfrac{1}{i}\pauli\cdot\nabla
    \psi}{\sigma_3(1-\mathbbm{1}_\Omega)\psi}&=iM \sum_{j=1}^2
  \int_{\R^2\setminus\Omega}  \partial_j (\sigma_j \psi(\bx),
  \sigma_3\psi(\bx)) d\bx\\
&=-iM\int_{\partial \Omega} {\bf L}\cdot \no d\omega,
\end{align*}
where ${\bf L}=(L_1,L_2)$ with $L_j(s)=(\sigma_j \psi(\vg(s)),
  \sigma_3\psi(\vg(s))) $ and $\no(s)=(\cos\alpha(s),\sin\alpha(s))$
  is the outward normal vector. Therefore, 
  \begin{equation}
    \label{eq:9}
\begin{split}
     \|H_M\psi\|^2&=\int_{\R^2}|\nabla \psi (\bx)|^2d\bx+
M^2 \int_{\R^2\setminus\Omega} |\psi (\bx)|^2d\bx\\
&\quad-2M\,{\rm Im}\big\{\int_0^L \overline{\psi_2(\vg(s))}
\psi_1(\vg(s)) e^{i\alpha(s)} ds\big\}.
\end{split}
\end{equation}
Using \eqref{projections} we  get that 
\begin{align*}
  2\,{\rm Im}\big\{\int_0^L \overline{\psi_2(\vg(s))}
\psi_1(\vg(s)) e^{i\alpha(s)} ds\big\} =\int_0^L\big[|P_+\psi(\vg(s))|^2-
|P_-\psi(\vg(s))|^2\big]  ds,
\end{align*}
which together \eqref{eq:9} implies \eqref{eq:11}.
\end{proof}
\begin{lemma}\label{thm2}
The operator $H_M$ has purely discrete spectrum between
$(-M,M)$. Moreover,  $E\in \sigma(H_M)$ if and only if $-E\in
   \sigma(H_M)$.
\end{lemma}
\begin{proof}
The operator 
$
  \tilde{H}_M:=\frac{1}{i} \pauli\cdot\nabla+\sigma_3M
$
has a spectral gap in  $(-M,M)$. Since the difference between $H_M$ and
$\tilde{H}_M$ is relatively compact with respect to $\tilde{H}_M$ we get that $\sigma_{\rm
  ess}(H_M)=\sigma_{\rm ess}(\tilde{H}_M)$. That the spectrum of $H_M$
is symmetric follows from the identity $UH_M=-H_MU$, where $U=\sigma_2\mathcal{C}$.
\end{proof}
\subsection{Energy estimates}
In the remainder of this work we frequently use the Trace Theorem
stated in the following form. For the proof see e.g.  \cite[Theorem
5.5.1]{evans}.
\begin{proposition}\label{tracetheorem}
For each $\varepsilon >0$ there is
a constant $C_\epsilon>0$ such that for all $\varphi\in
H^1(\Omega,\C^2)$
\begin{align}
  \label{eq:281}
  \|\varphi \|_{\partial\Omega}^2\le \varepsilon
  \|\nabla\varphi\|^2_\Omega+C_\epsilon\|\varphi\|^2_{\Omega}.
\end{align}
\end{proposition}
\begin{lemma}\label{qfm}
There exist constants $C,M_0>0$ such that for any 
$\psi\in H^1(\R^2,\C^2)$ and $M>M_0$ holds
\begin{equation}
  \label{eq:26}
\begin{split}
  \|H_M\psi\|^2&\ge \int_{\Omega} |\nabla \psi (\bx)|^2d\bx+
\frac{1}{2} \int_0^L |\psi(\vg(s))|^2 \alpha'(s) ds\\
&\qquad+
M\int_0^L |P_-\psi(\vg(s))|^2 ds-\frac{C}{M}\int_0^L  |\psi(\vg(s))|^2 ds.
\end{split}
\end{equation}
\end{lemma}
\begin{proof}
Using \eqref{eq:11} we have, for any $\psi\in H^1(\Omega,\C^2)$,
\begin{equation}\label{bi1}
\begin{split}
   \|H_M\psi\|^2&\ge \int_{\Omega} |\nabla \psi
   (\bx)|^2d\bx+M\int_0^L |P_-\psi(\vg(s))|^2 ds\\
&\quad +\int_{\R^2\setminus \Omega} |\nabla \psi
   (\bx)|^2d\bx+
M^2 \int_{\R^2\setminus\Omega} |\psi (\bx)|^2d\bx\\
&\quad-M\int_0^L |\psi(\vg(s))|^2 ds. 
\end{split}
\end{equation}
We now estimate the last three terms above. Recall the definition of
$Q_0$ from \eqref{q}.
Let $u,v:\R\to [0,1]\in C^2$ be a partition of unity in
$[0,\infty)$ with $u^2+v^2=1$ such that ${\rm supp}(u)\in [0,\delta)$
and $u(r)=1$ for $r\in [0,\delta/2]$. 
We define
\begin{align*}
   \psi_u(\bx)=\psi({\boldsymbol \kappa}(r,s)) u(r),\quad
  \psi_v(\bx)=\psi({\boldsymbol \kappa}(r,s)) v(r),\quad
  \bx={\boldsymbol \kappa}(r,s)\in Q_0 \cap \{\bx : \bx \notin \Omega\}.
\end{align*}
We further set  $\psi_v=\psi$ on
$\R^2\setminus (\Omega\cup Q_0)$.
We find by the IMS localization formula
\begin{align*}
 & \int_{\R^2\setminus \Omega} |\nabla \psi(\bx)|^2 d\bx\\
&\qquad\ge 
  \int_{\R^2\setminus \Omega}\big[|\nabla \psi_u (\bx)|^2 +|\nabla 
\psi_v (\bx)|^2 -c^2|\psi_u(\bx)|^2-c^2|\psi_v (\bx)|^2\big] d\bx,
\end{align*}
where $c=\max\{||\nabla u||_\infty, \|\nabla v\|_\infty\}$.
 Moreover, for $M>\sqrt{2}c$, we get
 \begin{align}
   \label{eq:15}
   \int_{\R^2\setminus \Omega} 
(M^2-c^2)|\psi_v (\bx)|^2 d\bx\ge 
 \frac{M^2}{2}\int_{\R^2\setminus \Omega}|\psi_v (\bx)|^2d\bx.
 \end{align}
Thus, 
\begin{equation}\label{bi2}
  \int_{\R^2\setminus \Omega} |\nabla \psi
   (\bx)|^2d\bx+
M^2 \int_{\R^2\setminus\Omega} |\psi (\bx)|^2d\bx\ge F[\psi_u]+\frac{M^2}{2}\int_{\R^2\setminus \Omega}|\psi_v (\bx)|^2d\bx,
\end{equation}
where
\begin{align*}
  F[\psi_u]:=\int_{\R^2\setminus \Omega} |\nabla \psi_u
   (\bx)|^2d\bx +(M^2-c^2) \int_{\R^2\setminus\Omega} |\psi_u (\bx)|^2d\bx.
\end{align*}
Our next goal is to estimate $F[\psi_u]$. Clearly,
\begin{align*}
  F[\psi_u]\ge \int_0^\delta \int_0^L \big[|{\partial_r}\psi_u({\boldsymbol \kappa}(r,s))|^2+(M^2-c^2) |\psi_u({\boldsymbol \kappa}(r,s))|^2)
  \big](1+r\alpha'(s)) ds dr.
\end{align*}
In order to estimate the integral in $dr$ above we
apply Lemma~\ref{variation} (from Appendix~\ref{functional}) with $k=\sqrt{M^2-c^2}$, $\beta=\alpha'(s)$
and $f=\psi_u({\boldsymbol \kappa}(\cdot,s))$.  To this end we set  
\begin{equation}\label{is}
I(s):=\int_0^\delta
|\psi_u({\boldsymbol \kappa}(r,s))|^2dr
\end{equation}
and define (compare with \eqref{RR}) 
\begin{equation}
  \label{eq:49}
  \mathcal{R}(s):=R[\psi_u({\boldsymbol \kappa}(\cdot,s))]=
\frac{M^2-c^2}{16} I(s) \mathbbm{1}_K(s),
\end{equation}
where the set $K\subset \R^2$ is defined as  
\begin{equation}\label{esssupp}
K:={\rm supp}\,\big[{\rm max}\{I(s)-2|\psi(\vg(s))|^2/\sqrt{M^2-c^2},0\}\big].
\end{equation}
From Lemma~\ref{variation} we get using $\psi_u({\boldsymbol
  \kappa}(0,s)=\psi(\vg(s))$ that for $M$ sufficiently large
\begin{align*}
\begin{split}
 F[\psi_u]&\ge \sqrt{M^2-c^2}\|\psi\|_{\partial \Omega}^2+\frac{1}{2}
  \int_0^L
  |\psi(\vg(s))|^2 \alpha'(s)ds
% \\
%   &\quad
+\int_0^L  \mathcal{R}(s) ds+\mathcal{O}(\tfrac{1}{M}) \|\psi\|_{\partial \Omega}^2\\
  &= M \|\psi\|_{\partial \Omega}^2+\frac{1}{2} \int_0^L
  |\psi(\vg(s))|^2 \alpha'(s)ds
%\\&\quad
+\int_0^L  \mathcal{R}(s) ds+\mathcal{O}(\tfrac{1}{M})\|\psi\|_{\partial \Omega}^2. \\
  \end{split}
\end{align*}
 Thus, combining the latter estimate with \eqref{bi1}
and \eqref{bi2} yields
\begin{equation}
  \label{eq:26-a}
\begin{split}
  \|H_M\psi\|^2&\ge \int_{\Omega} |\nabla \psi (\bx)|^2d\bx+
\frac{1}{2} \int_0^L |\psi(\vg(s))|^2 \alpha'(s) ds+
M\|P_-\psi\|_{\partial\Omega}^2\\
&\quad+\frac{M^2}{2}\int_{\R^2\setminus \Omega}|\psi_v (\bx)|^2d\bx
+\int_0^L \mathcal{R}(s) ds+\mathcal{O}(\tfrac{1}{M}) \|\psi\|_{\partial \Omega}^2.
\end{split}
\end{equation}
Thus, we obtain \eqref{eq:26} dropping the fourth and
fifth term on the right hand  side  of \eqref{eq:26-a}. 
\end{proof}
%%%%%%%
%%%%%%%
Note that in the above proof we did not use the full strength of
\eqref{eq:26-a}. However, we do so in the proof of the next lemma.
%%%%%%%
\begin{lemma}\label{coro1.1}
  For any $A>0$ there are constants $C,M_0>0$ such that for any
  $\psi\in E_{(-A,A)} (H_M)L^2(\R^2,\C^2) $ and any $M>M_0$
\begin{align}
\label{egal}
\|\psi\|^2_{H^1(\Omega)} &\le C\|\psi\|^2,\\
\label{finitetr}
 \|\psi\|_{\partial\Omega}^2 &\le C\|\psi\|^2,\\
\label{eq:28.1}
\|\psi\|^2_{\R^2\setminus \Omega}&\le\frac{C}{M}\|\psi\|^2.
\end{align}
\end{lemma}
\begin{proof}
Note that \eqref{finitetr} is a consequence of \eqref{egal} and \eqref{eq:281}.
According to  \eqref{eq:26-a} we have,  for sufficiently large $M>0$, 
 \begin{equation}
  \label{22}
\begin{split}
  \|H_M\psi\|^2&\ge \int_{\Omega} |\nabla \psi (\bx)|^2d\bx
-\|\alpha'\|_\infty \|\psi\|_{\partial\Omega}^2
+\frac{M^2}{2}\|\psi_v\|^2.\\
\end{split}
\end{equation}
Using  \eqref{eq:281} we get, for some $c_1>0$,
\begin{equation}
  \label{eq:30}
  \|H_M\psi\|^2\ge \frac12\|\nabla \psi \|_\Omega^2+\frac{M^2}{2}
  \|\psi_v\|^2-c_1\|\psi \|^2.
\end{equation}
 Since $ \|H_M\psi\|\le |A|\|\psi\|$  we obtain
 \eqref{egal}. Moreover, using again \eqref{eq:30} we get for some
 constant $c_2>0$ 
 \begin{align}
   \label{eq:25}
    \|\psi_v\|^2\le \frac{c_2}{M^2} \|\psi\|^2.
 \end{align} 
 In order to prove \eqref{eq:28.1} it suffices to show that $\|
 \psi_u\|_{\R^2\setminus \Omega}^2\le C\|\psi\|^2/M$. First note that
since $\delta<1/\|\alpha'\|_\infty$ (see \eqref{q})
\begin{align}\label{gggg} 
  \|\psi_u\|_{\R^2\setminus \Omega}^2=\int_0^L \!\!\!\!\int_0^\delta |\psi_u
  ({\boldsymbol \kappa}(r,s))|^2 (1+r\alpha'(s)) ds dr
\le 2\int_0^L I(s) ds,
\end{align}
where $I(s)$ is defined in \eqref{is}.
Using \eqref{eq:281} as above  we  get from  
\eqref{eq:26-a} that
\begin{align}
  \label{eq:31}
 |A|^2\|\psi\|^2\ge \|H_M\psi\|^2\ge \int_0^L 
\mathcal{R}(s) ds-c_1\|\psi \|^2.
\end{align}
This together with \eqref{eq:49}  and \eqref{esssupp} imply that
\begin{align}
  \label{eq:32}
  \int_0^L I(s)\mathbbm{1}_K(s)  ds
  \le \frac{16 (c_1+|A|^2)}{M^2-c^2} \|\psi\|^2.
\end{align}
Using again  the definition of $K$  \eqref{esssupp} and
\eqref{finitetr} we further obtain that
\begin{align*}
   \int_0^L I(s) (1-\mathbbm{1}_K(s)) ds\le
  \frac{2}{\sqrt{M^2-c^2} }\int_0^L |\psi(\vg(s))|^2 ds\le
  \frac{2 C \|\psi\|^2}{\sqrt{M^2-c^2} }.
\end{align*}
Thus, combining the latter inequality with \eqref{eq:32} and
\eqref{gggg} we obtain \eqref{eq:28.1}.
\end{proof}
%%%%%%%
%%%%%%%
%%%%%%%
\begin{corollary}\label{coro1}
For any $A>0$ there are constants $C,M_0>0$
 such that
for any $\psi\in  E_{(-A,A)} (H_M)L^2(\R^2,\C^2) $ and any $M>M_0$
\begin{align}
& \label{eq:29.1}
\|\psi\|^2_{H^1(\Omega)} \le C\|\psi\|^2_{\Omega},\\
  \label{eq:27}
  &\|H_M\psi\|^2\ge \int_{\Omega} |\nabla \psi (\bx)|^2d\bx+
\int_{\partial\Omega} |\psi(\vg(s))|^2
\alpha'(s)\frac{ds}{2}-\frac{C}{M}\|\psi\|^2,\\
\label{eq:29}
&\|P_-\psi\|_{\partial\Omega}^2 \le\frac{C}{M}\|\psi\|^2.
\end{align}
\end{corollary}
\begin{proof}
 The estimate
  \eqref{eq:29.1} follows from \eqref{egal} and \eqref{eq:28.1}.
 From \eqref{finitetr} and \eqref{eq:26} we obtain
  \eqref{eq:27}. Finally \eqref{eq:29} is a consequence of
  \eqref{eq:26}, \eqref{finitetr}, and the fact that $\|H_M\psi\|\le
  |A|\|\psi\|$.
\end{proof}
\section{Proof of the main theorem}
\begin{lemma}\label{L1}
For $A\notin \sigma(H_\infty)$ assume that
\begin{align*}
  {\rm dim Ran} E_{(-A,A)}(H_\infty) =N.
\end{align*}
Then 
there is $M_0>0$ such that for all $M>M_0$ we find $\mathcal{L}_N \equiv \mathcal{L}_N(M)\subset
H^1(\R^2,\C^2)$ with $ {\rm dim}\mathcal{L}_N=N $ and  
\begin{align*}
  \|H_M\varphi\|^2< A^2 \|\varphi\|^2,\quad \varphi\in \mathcal{L}_N.
\end{align*}
I.e., $H_M$ has at least $N$ eigenvalues in $(-A,A)$ for
all $ M>M_0$.
\end{lemma}
\begin{proof}
  Recall the definition of $Q_0$ from \eqref{q}.  Let  
\begin{align}
    \label{eq:23}
    c_1:=\max\{\lambda\in [0,A) : \lambda \,\,\mbox{is an eigenvalue
      of }\,\, H_\infty\}.
  \end{align}
 For any $\varphi\in
  \mathcal{M}_N:={\rm Ran} E_{(-A,A)}(H_\infty)$ normalized we define
 \begin{align}
   \label{eq:12}
   \psi_M(\bx):=\left\{\begin{array}{cc}
\varphi(\bx),&\quad \bx\in \Omega\\
\varphi(\vg(s)) e^{-Mr}\zeta(r),&\quad \bx=
{\boldsymbol \kappa}(r,s)\in Q_0\cap\{\bx:\bx\notin \Omega\}\\
 0,&\quad \bx\notin Q_0\cup \Omega
     \end{array}    
\right.,
 \end{align}
where $\zeta\in C^2(\R, [0,1])$ with $\zeta(r)=1$ for $r\in[0,\delta/2]$ and
vanishes for $r>\delta$. We denote by 
 $ \mathcal{L}_N$ the linear subspace of all such $\psi_M$ with 
$\varphi\in\mathcal{M}_N$. Clearly, 
${\rm dim}\,\mathcal{L}_N={\rm dim}\,\mathcal{M}_N=N$.
Since $\varphi\in H^2(\Omega,\C^2)$, by
Theorem \ref{thm1}, we see that $\varphi(\vg(s)), \varphi(\vg(s))'\in
L^2(\partial\Omega,\C^2)$. In particular,
$\mathcal{L}_N \subset H^1(\R^2,\C^2)$. 
Let $(\varphi_j)_{j\in\N}$ be a basis in $\mathcal{M}_N$ orthonormal
in the $L^2$ sense. 
Defining 
$$\beta_N:=\max_{j=1,\dots,N}\|\varphi_j\|_{H^1(\partial\Omega,\C^2)}^2,$$
we get, for any normalized $\varphi\in \mathcal{M}_N$,
\begin{align}
  \label{eq:13}
  \|\varphi\|^2_{H^1(\partial\Omega, \C^2)}\le N\beta_N.
\end{align}
Let $\psi_M\in \mathcal{L}_N$ with $\psi_M
\mathbbm{1}_\Omega=\varphi\in \mathcal{M}_N$ be such that
$\|\varphi\|=\|\psi_M\|_\Omega=1$. We first show that
\begin{align}
  \label{eq:14}
  \|\psi_M\|^2=1+\mathcal{O}(1/M)\quad \mbox{as} \quad M\to \infty.
\end{align}
This follows from the estimate
\begin{align*}
   |\|\psi_M\|^2-1|\le 2\nn\varphi\nn^2_{\partial \Omega} \int_0^\infty
e^{-2Mr} dr
&=\nn\varphi\nn^2_{\partial \Omega}/M\le N\beta_N/M, 
 \end{align*}
since
\begin{align*}
  \|\psi_M\|^2=1+
\int_0^\delta \int_0^L |\varphi(\vg(s))|^2e^{-2Mr} \zeta^2(r)
(1+\alpha(s)' r) ds dr.
\end{align*}
Next we estimate the $\|H_M\psi_M\|^2$. Using \eqref{eq:11} (see also 
\ref{grads}) we get
\begin{align*}
  \|H_M\psi_M\|^2&=\|\nabla\varphi\|_\Omega^2+M^2\int_0^\delta \int_0^L |\varphi(\vg(s))|^2e^{-2Mr} \zeta^2(r)
(1+\alpha(s)' r) ds dr\\
&\quad +\int_0^\delta \int_0^L \big[|\partial_s \varphi(\vg(s))e^{-Mr}
\zeta(r)|^2\big] (1+\alpha(s)' r)^{-1} ds dr\\
&\quad +\int_0^\delta \int_0^L \big[|\partial_r \varphi(\vg(s))e^{-Mr}
\zeta(r)|^2\big] (1+\alpha(s)' r) ds dr-M\nn\varphi\nn_{\partial\Omega}^2\\
&=: \|\nabla\varphi\|_\Omega^2+I_1+I_2+I_3-M\nn\varphi\nn_{\partial\Omega}^2,
\end{align*}
where we used that $P_+(s)\varphi(\vg(s))=\varphi(\vg(s))$ for
$s\in[0,L]$. We estimate the terms in the right hand side of the
previous equality. For $I_2$ we have
\begin{align*}
  I_2\le cN\beta_N \int_0^\infty
e^{-2Mr} dr=\mathcal{O}(1/M),
\end{align*}
for some positive constant $c$.
Furthermore, 
\begin{align*}
  I_3&=M^2\int_0^\delta \int_0^L  |\varphi(\vg(s)) e^{-Mr} \zeta(r)|^2
  (1+\alpha(s)' r) ds dr\\
&\quad-2M\int_0^\delta \int_0^L  |\varphi(\vg(s))
  e^{-Mr}|^2 \zeta'(r) \zeta(r) (1+\alpha(s)' r) ds dr\\
&\quad +\int_0^\delta \int_0^L  |\varphi(\vg(s)) e^{-Mr} \zeta'(r)|^2
  (1+\alpha(s)' r) ds dr\\
&=: I_{3,1}+I_{3,2}+I_{3,3}.
\end{align*}
Using that $\zeta'(r)=0$ for $r\in[0,\delta/2]$ we get
\begin{align*}
  I_{3,2}\le 4M \nn\varphi\nn^2_{\partial\Omega} \|\zeta'\|_\infty
  e^{-M\delta} \delta=\mathcal{O}(M e^{-M\delta}).
\end{align*}
Similarly, we see that 
$$I_{3,3}=\mathcal{O}(e^{-M\delta}).$$
Noting that $ I_{3,1}=I_1$ we have altogether
\begin{align}
  \label{eq:17}
   \|H_M\psi_M\|^2&=\|\nabla\varphi\|_\Omega^2+2I_1-M\nn\varphi\nn_{\partial\Omega}^2+\mathcal{O}(1/M).
\end{align}
Finally we estimate $I_1$ as follows 
\begin{align*}
 2I_1
&=2M^2\nn\varphi\nn_{\partial\Omega}^2 \int_0^\delta e^{-2Mr}\zeta^2(r)
dr\\
&\quad+2M^2\int_0^L  |\varphi(\vg(s))|^2\alpha'(s)ds \cdot \int_0^\delta
e^{-2Mr}\zeta^2(r) r
dr\\
&\le  M\nn\varphi\nn_{\partial\Omega}^2+2M^2 \int_0^\delta
e^{-2Mr}\zeta^2(r) r dr\cdot\int_0^L  |\varphi(\vg(s))|^2\alpha'(s)ds.
\end{align*}
In addition we have that 
\begin{align*}
  &\int_0^\delta e^{-2Mr}\zeta^2(r) r dr= \int_0^{\delta/2} e^{-2Mr} r
  dr+
  \int_{\delta/2}^\delta e^{-2Mr}\zeta^2(r) r dr\\
  &=\frac{1}{4M^2}\int_0^\infty e^{-u} u du-\frac{1}{4M^2}
  \int_{M\delta}^\infty e^{-u} u du + \int_{\delta/2}^\delta
  e^{-2Mr}\zeta^2(r) r dr
\\&
=\frac{1}{4M^2}+\mathcal{O}(e^{-M\delta/2}).
\end{align*}
This implies that
\begin{align}
  \label{eq:19}
  2I_1\le  M\nn\varphi\nn_{\partial\Omega}^2+\frac12 \int_0^L
  |\varphi(\vg(s))|^2\alpha'(s)ds +\mathcal{O}(M^2 e^{-M\delta/2}).
\end{align}
Therefore, according to \eqref{eq:17}, we find
\begin{align}
  \label{eq:20}
    \|H_M\psi_M\|^2&=\|\nabla\varphi\|_\Omega^2+\frac12 \int_0^L
    |\varphi(\vg(s))|^2\alpha'(s)ds+\mathcal{O}(1/M)\le c_1^2+\mathcal{O}(1/M),
\end{align}
where in the last inequality we use \eqref{qf} and \eqref{eq:23}. This
together with \eqref{eq:14} implies that
\begin{align}
  \label{eq:21}
    \|H_M\psi_M\|^2/\|\psi_M\|^2\le c_1^2+\mathcal{O}(1/M), \quad
    \psi_M\in \mathcal{L}_N.
\end{align}
Since $c_1<A$ we get, by the Spectral Theorem, that $\displaystyle 
{\rm dim Ran}_{(-A,A)}(H_M)\ge N$.
\end{proof}
\begin{lemma}\label{comparation}
Let $0<A\notin \sigma(H_\infty) $ be fixed. Then there is a constant
$M_0>0$ such that
\begin{align}
  \label{eq:42}
  {\rm dim Ran}
E_{(-A,A)} (H_\infty)=
{\rm dim Ran}
E_{(-A,A)} (H_\M),
\end{align}
for any $M>M_0$. 
\end{lemma}
\begin{proof}
Let $N:= {\rm dim Ran}
E_{(-A,A)} (H_\infty)$. That
\begin{align*}
    {\rm dim Ran}
E_{(-A,A)} (H_\M)\ge N
\end{align*}
follows from Lemma \ref{L1}, for large $M>0$.  Assume that the reverse
inequality does not hold. Due to our assumption
there exists a sequence $(M_j)_{j\in \N}$ with $M_j\to \infty$ such that 
\begin{align*}
{\rm dim Ran}
E_{(-A,A)} (H_{\M_j})\ge N+1.
\end{align*}
Hence we can find a normalized function $\psi_j\in {\rm Ran}E_{(-A,A)}
(H_{\M_j})$ which is orthogonal to the eigenfunctions of $H_\infty$
with eigenvalues in $(-A,A)$ (extended by zero in $\R^2\setminus
\Omega$). Define $\varphi_j:=\psi_j \chi_\Omega$. Due to equation
\eqref{eq:28.1}, \eqref{eq:29}, and \eqref{eq:29.1} we have that
$\|\varphi_j\|_{H^1(\Omega)}$ is bounded uniformly in $M_j$ and,
moreover, as $j\to \infty$
\begin{align*}
 \|\varphi_j\|_{\Omega}\to 1\quad\mbox{and}\quad 
\|P_-\varphi_j\|_{\partial\Omega}\to 0. 
 \end{align*}
In particular, by the Theorem of Banach-Alaoglu, the sequence $(\varphi_j)_{j\in\N}$ contains a subsequence (also
called $(\varphi_j)$) such that, as $j\to\infty$,
\begin{align}
  \label{eq:45}
  \varphi_j\rightharpoonup : \varphi\quad \mbox{in} \quad H^1(\Omega,\C^2),
\end{align}
with 
\begin{align}\label{ww2}
 & \|\varphi\|_ {H^1(\Omega,\C^2)}\le \liminf_{j\to\infty}
 \|\varphi_j\|_{H^1(\Omega,\C^2)}.
\end{align}
In addition, using the Theorem of Rellich-Kondrachov, see \cite[Theorem 6.2 (4)]{adams},
\begin{align}\label{ww1}
&\|\varphi_j -\varphi\|_\Omega\to 0,\qquad \|\varphi_j
  -\varphi\|_{\partial\Omega}\to 0,\quad j\to \infty,
\end{align}
which implies that
\begin{align*}
 \|\varphi\|_\Omega=1\quad \|P_-\varphi\|_{\partial\Omega}=0.
\end{align*}
Therefore, $\varphi\in\mathcal{D}(H_\infty)$ and satisfies $\varphi
\perp {\rm Ran} E_{(-A,A)} (H_\infty)$.  Let $\lambda_n>0$ the
largest eigenvalue of $H_\infty$ in $(-A,A)$ and $\lambda_{n+1}>A$ be
the next positive eigenvalue of $H_\infty$. Define
$\varepsilon=(\lambda_{n+1}^2-A^2)/2$. Then, for $j$ large
enough, we have in view of \eqref{ww2} and \eqref{ww1} that
\begin{align*}
  \|H_\infty\varphi\|^2&=
 \|\varphi\|^2_{H^1(\Omega,\C^2)}-\|\varphi\|^2+
\frac{1}{2}\int_{\partial\Omega} |\varphi(\vg(s))|^2
\alpha'(s) ds \\
&\le   \|\varphi_j\|^2_{H^1(\Omega,\C^2)}-\|\varphi_j\|^2+
\frac12\int_{\partial\Omega} |\psi_j(\vg(s))|^2
\alpha'(s)ds +\frac{\ve}{2}\\
&= \int_\Omega |\nabla\psi_j(\bx)|^2 d\bx +\frac12\int_{\partial\Omega} |\psi_j(\vg(s))|^2
\alpha'(s)ds +\frac{\ve}{2}\\
&\le \|H_M\psi_j\|^2+ \varepsilon,
\end{align*}
where in the last inequality we used \eqref{eq:27}.
The last inequality contradicts the assumption that $\varphi
\perp {\rm Ran} E_{(-A,A)} (H_\infty)$ since $ \|H_M\psi_j\|^2+ \varepsilon<\lambda_{n+1}^2$.
\end{proof}
\begin{corollary}\label{convergenceev}
As $M\to \infty$ the eigenvalues of $H_M$ convege uniformly on each
bounded spectral interval $(A,B)$ against the eigenvalues of
$H_\infty$. More precisely, if $\lambda_j$ and  $\lambda_{j+1}$ are
two subsequent eigenvalues of $H_\infty$,
$\lambda_{j+1}>\lambda_{j}$, then for each $\varepsilon>0$ there
exists $M_0>0$ such that for all $M>M_0$ 
\begin{align*}
  {\rm dim Ran}
E_{\lambda_{j}} (H_\infty)=
{\rm dim Ran}
E_{ (\lambda_{j}-\varepsilon, \lambda_{j}+\varepsilon)} (H_\M),
\end{align*}
and 
\begin{align*}
  E_{ (\lambda_{j}+\varepsilon, \lambda_{j+1}-\varepsilon)} (H_\M)=\emptyset.
\end{align*}
\end{corollary}
\begin{proof}
This follows from Lemma \ref{comparation} and the symmetry of the spectra
of $H_M$ and $H_\infty$. 
\end{proof}\label{lemma6}
\begin{lemma}\label{lemmalast}
Let $0<A\notin \sigma(H_\infty)$ be fixed. Then, 
  \begin{align*}
   \big\| \widetilde{E}_{(-A,A)}(H_\infty)- E_{(-A,A)}(H_M) \big\|\to
   0,\quad \mbox{as}\quad  M\to \infty.
  \end{align*}
  Here $ \widetilde{E} (H_\infty)={E}
  (H_\infty)\oplus \{0\}$ with respect to the splitting
  $\mathcal{H}=L^2(\Omega,\C^2)\oplus L^2(\R^2\setminus
  \Omega,\C^2)$.
\end{lemma}
\begin{proof}
Let $N:={\rm dim Ran} \widetilde{E}_{(-A,A)}(H_\infty)$ and 
let $\varphi_j$, $j=1,2,\dots,N,$ be an orthonormal  basis on the
range of $\widetilde{E}_{(-A,A)}(H_\infty)$. For each $\varphi_j$ we define $\psi_j^M$ according to 
\eqref{eq:12}. Due to \eqref{eq:14} we have, for $M>1$ large enough, that 
\begin{equation}\label{difference}
\|\varphi_j-\psi_j^M\|=\mathcal{O}(M^{-1/2}).
\end{equation}
In addition,
\begin{equation}
  \label{eq:47}
  \sps{\psi_j^M}{\psi_k^M}=\mathcal{O}(M^{-1}),\quad j\not= k.
\end{equation}
Let 
\begin{equation}
  \label{eq:35}
  \mathcal{L}^M:= {\rm span}\{\psi_1^M,\dots,\psi_N^M\}.
\end{equation}
Clearly ${\rm dim}\,\mathcal{L}^M=N$. Let $P^M$ be the orthogonal
projection onto $\mathcal{L}^M\subset L^2(\R^2,\C^2)$.  Due to \eqref{difference} holds 
that 
\begin{align}
  \label{eq:48}
  \| \widetilde{E}_{(-A,A)}(H_\infty)-P^M\|\to 0, \quad \mbox{as}\quad
  M\to \infty.
\end{align}
Next we show that $\|P^M- E_{(-A,A)}(H_M) \|\to 0$ as
$M\to \infty.$ Let $0\le
|\lambda_1|<|\lambda_2|<\dots<|\lambda_j|<\dots$ be the absolute
values of the eigenvalues of $H_\infty$. (We allow these eigenvalues
to be degenerate.) Define $(A_j)_{j\in\N}$ as
$A_j:=(|\lambda_j|+|\lambda_{j+1}|)/2$. Let $\varphi_1,
\dots,\varphi_{p_j}$ be  an orthonormal basis of eigenfunctions of $H_\infty$ on the
range of $E_{(-A_j,A_j)}(H_\infty)$. Using the functions $\varphi_1,
\dots,\varphi_{p_j}$ we construct $\psi^M_1,
\dots,\psi^M_{p_j}$ as in \eqref{eq:12}
and denote by $P_j^M$ the orthogonal projection onto the
$\rm{span}\{\psi_1^M,\dots,\psi_{p_j}^M\}$.

 We now show, using induction,  that 
\begin{align}\label{prodiff}
\|P_j^M-
E_{(-A_j,A_j)}(H_M) \|\to 0\quad\mbox{as}\quad M\to \infty,
\end{align} 
for each $j<N$, where
$N>0$ is some arbitrary fixed number. 
We set $I_j:=(-A_j,A_j)$
\begin{align*}
  \mu_j(M):=\min\{|\lambda|\,|\, \lambda \in \sigma(H_M) \cap I_j\setminus I_{j-1}\},
\end{align*}
with the convention  $I_0=\emptyset$. Notice that $\mu_j(M)\to
|\lambda_j|$ as $M\to \infty$. Due to Corollary
\ref{convergenceev} we may assume that $M>0$ is so large that, for
all $j<N$,
\begin{align*}
  {\rm dim Ran}
P_j^M=
{\rm dim Ran}
E_{I_j} (H_\M).
\end{align*}
Then, according to \cite[Theorem I.6.34]{kato}, the norm in
\eqref{prodiff} equals the norms  $\|(1-E_{I_j} (H_\M))
P_j^M\|=\|E_{I_j} (H_\M)(1-P_j^M)\|$.\\
\noindent
{\it Induction start}: We write $P^\perp:=1-P$ for an orthogonal
projection $P$. Let $\psi\in
P_1^M L^2(\R^2,\C^2)$. By the spectral
theorem we have 
\begin{align}
  \label{eq:50}
  \|H_M\psi\|^2\ge \mu_1^2(M)\|E_{ I_1} (H_\M)\psi\|^2+A_1^2\|E_{ I_1} (H_M)^\perp\psi\|^2.\end{align}
According to \eqref{eq:21} we have 
\begin{align*}
    \|H_M\psi\|^2&\le \lambda_1^2\|\psi\|^2+
    \mathcal{O}(1/M)\|\psi\|^2\\&
=
 \lambda_1^2\|E_{ I_1} (H_\M)\psi\|^2+ \lambda_1^2\|E_{ I_1}
 (H_\M)^\perp\psi\|^2
+\mathcal{O}(1/M)\|\psi\|^2.
\end{align*}
A combination of the above inequality and \eqref{eq:50} together with
the fact that $\mu_1(M)\to |\lambda_1|$ as $M\to \infty$ implies that
$\|E_{ I_1} (H_M)^\perp P^M_1\|\to 0$ as $M\to \infty$. \\
\noindent
{\it Induction step}: Assume that the statement holds for the interval
$I_j$. Let $\varepsilon>0$ and $\psi\in
(1-P_{j}^M) P_{j+1}^M L^2(\R^2,\C^2)$.  We have
\begin{align}
  \label{eq:51}
  \|H_M\psi\|^2\ge \|H_M E_{
  I_{j+1}} (H_M)\psi\|^2+A_{j+1}^2\|E_{
  I_{j+1}} (H_M)^\perp\psi\|^2.
\end{align}
Using that $\psi=(1-P^M_j)\psi$ we find that 
$$\|H_M E_{I_{j+1}} (H_M)\psi\|\ge \|H_M E_{I_{j+1}\setminus I_j}(H_M)\psi\|-|A_j|\|E_{I_j}(H_M)
(1-P^M_j) \| \|\psi\|. $$
The last term above converges to zero as $M\to \infty$ due to  the
induction hypothesis. Therefore, we get for
sufficiently large $M$,
\begin{align}
  \label{eq:43}
  \|H_M\psi\|^2\ge \mu_{j+1}(M)^2\|
E_{I_{j+1}\setminus I_j} (H_M)\psi\|^2+A_{j+1}^2\|E_{
  I_{j+1}} (H_M)^\perp\psi\|^2-\varepsilon\|\psi\|^2.
\end{align}
Using \eqref{eq:21} we obtain
\begin{align*}
  \|H_M\psi\|^2&\le
  \lambda_{j+1}^2\|\psi\|^2+\mathcal{O}(1/M)\|\psi\|^2\\
&\le \lambda_{j+1}^2 \|E_{I_{j+1}\setminus I_j}
(H_M)\psi\|^2+\lambda_{j+1}^2 
\|E_{I_{j+1}} (H_M) ^\perp\psi\|^2+(\varepsilon+\mathcal{O}(1/M))\|\psi\|^2.
\end{align*}
A combination of this with \eqref{eq:43} gives that $\|E_{
  I_{j+1}} (H_M)^\perp\psi\|/\|\psi\|\to 0$, since $\mu_{j+1}(M)\to |\lambda_{j+1}|$. 
From this follows that 
\begin{align}\label{persona}
\| E_{I_{j+1}} (H_M)^\perp (1-P^M_{j})P_{j+1}^M\|\to 0,\quad\mbox{as} \quad j\to 0.
\end{align}
Finally, the above equation \eqref{persona}, the identity
\begin{align*}
  E_{I_{j+1}} (H_M)^\perp P^M_{j+1} = 
E_{I_{j+1}} (H_M)^\perp E_{I_{j}} (H_M)^\perp  P^M_{j}+ E_{I_{j+1}} (H_M)^\perp (1-P^M_{j})P^M_{j+1},
\end{align*}
and the induction hypothesis  imply the claim.
\end{proof}
%%%%%%
%%%%%%
%%%%%%
\begin{lemma}\label{lastlemma}
Let $\lambda \in \sigma(H_\infty)$. Then, for any $\varepsilon>0$, holds 
\begin{align}\label{ecuac}
\|\widetilde{E}_{\{\lambda\}} (H_\infty)-
{E}_{ (\lambda-\varepsilon, \lambda+\varepsilon)} (H_\M)\|\to
0\quad\mbox{as}\quad M\to\infty,
\end{align}
where $\widetilde{E}_{\{\lambda\}}
(H_\infty)$ is defined as in Lemma \ref{lemmalast}.
\end{lemma}
\begin{proof}
  We show the statement by contradiction.  Recall that $\lambda\not=0$
  \cite{BFSV2016}. If \eqref{ecuac} does not hold there exists an
  eigenfunction $\phi_M$ of $H_M$ with eigenvalue $\lambda_M$
  belonging to the range of $ (\lambda-\varepsilon,
  \lambda+\varepsilon)$ such that $\|\phi_M-\varphi\|\to 0$, $M\to
  \infty$, where $\varphi$ is an eigenfunction of $H_\infty$ with
  eigenvalue $-\lambda$. Using the test function $\psi_M$ constructed
  from $\varphi$ as in \eqref{eq:12} we have
\begin{align*}
  \sps{\phi_M}{\varphi}&=-\lambda^{-1} \sps{\phi_M}{H_\infty \varphi}=
-\lambda^{-1} \int_\Omega \overline{\phi_M(\bx)}  [H_M\psi_M](\bx) d\bx
\\
&=-\lambda^{-1} \int_{\R^2} \overline{\phi_M(\bx) } [H_M\psi_M](\bx) d\bx
+\lambda^{-1} \int_{\R^2\setminus \Omega} \overline{\phi_M(\bx)}
[H_M\psi_M](\bx) d\bx\\
&\le -\frac{\lambda_M}{\lambda}\sps{\phi_M}{\psi_M}+|\lambda^{-1}|
\|H_M\psi_M\| \|\phi_M\|_{\R^2\setminus \Omega}.
\end{align*}
Since $\|H_M\psi_M\| $ is uniformly bounded by \eqref{eq:21},
$\|\phi_M\|_{\R^2\setminus \Omega}$ and $\|\varphi-\psi_M\|$ converge to zero, and $\lambda_M\to \lambda$ we find a
contradiction when taking the limit $M\to \infty$ of the above inequality.
\end{proof}
%%%%%%%%%%% 
%%%%%%%%%%%%%%%%%%%%%%%%%%%%%%%%%%%%%%%%%%%%%%%%%%%%%%%%%%%%%%%%%%%%%%%%%%%%%%%%
%%%%%%%%%%%%%%%%%%%%%%%%%%%%%%%%%%%%%%%%%%%%%%%%%%%%%%%%%%%%%%%%%%%%%%%%%%%%%%%%
\noindent
{\bf Acknowledgments.}
E.S thanks Reinhold Egger, Alessandro de Martino, and Heinz Siedentop
for stimulating discussions.
 S.W. acknowledge support from SFB-TR12 ``Symmetries and
Universality in Mesoscopic Systems" and SFB 1173 of the DFG.  E.S. has been
supported by Fondecyt (Chile) project 1141008 and Iniciativa
Cient\'ifica Milenio (Chile) through the Millenium Nucleus RC–120002
``F\'isica Matem\'atica'' .
%%%%%%%%%%%
\begin{appendix}\label{appendix}
\section{Proof of Lemma \ref{energyi}}\label{appendixa1}
\begin{lemma}\label{limitarg}
  For any $\varphi\in \mathcal{D}_\infty$ there is a sequence
  $(\varphi_n)_{n\in\N}\subset C^1(\overline{\Omega}, \C^2)$ with
  $\varphi_n\to \varphi$ in the $H^1$-norm, such that $\varphi_n$
  satisfies the boundary conditions, i.e., $P_-\varphi_n=0$ for all
  $n\in\N$.
\end{lemma}
\begin{proof}
Recall the definition of $Q_0$ in \eqref{q}.
  Let  $\chi\in C^\infty((-\delta, 0], [0,1])$ be a smooth characteristic function
 with $\chi(r)=0$ for $r\in (-\delta, -\delta/2)$ and $\chi(r)=1$ for 
 $r\in (-\delta/4, 0]$. We define the function
 $\tilde{\alpha}:\overline{\Omega}\to\R$ 
\begin{align*}
\tilde{\alpha}(\bx)=\chi(r)\alpha(s) \quad\mbox{for}\quad
 \bx={\boldsymbol \kappa}(r,s)\in
Q_0\cap \overline{\Omega},
\end{align*}
and being zero otherwise.

For $\varphi\in\mathcal{D}_\infty$ we define 
\begin{align*}
  &\psi_1:=\varphi_2-ie^{i\tilde\alpha}\varphi_1\,\,\in H^1_0(\Omega),\\
& \psi_2:=\varphi_2+ie^{i\tilde\alpha}\varphi_1\,\,\in H^1(\Omega).
\end{align*}
There exist sequences $(\psi_{1,n})_{n\in\N}\subset
C^\infty_0(\Omega)$ and $(\psi_{2,n})_{n\in\N}\subset
C^\infty(\overline{\Omega})$ converging to $\psi_1$ and $\psi_2$ in
the $H^1$-norm, respectively. We define further, for $n\in\N$, the
following $C^1$-functions 
\begin{align*}
  \varphi_{1,n}:=\frac{e^{-i\tilde\alpha}}{2i}
  (\psi_{2,n}-\psi_{1,n}),
\qquad 
\varphi_{2,n}:=\frac{1}{2}
  (\psi_{1,n}+\psi_{2,n}).
\end{align*}
Clearly $\varphi_{1,n}$ and $\varphi_{2,n}$ converge to $\varphi_1$
and $\varphi_2$ in the $H^1$-norm, respectively. Moreover, since
$\psi_{1,n}\!\!\!\upharpoonright\!\partial\Omega=0$, one easily verifies
that $\varphi_n:=(\varphi_{1,n},\varphi_{2,n})^{\rm T}$ satisfies the
boundary conditions.
\end{proof}
\begin{proof}[Proof of Lemma \ref{energyi}]
We  compute, for $\varphi\in  C^1(\overline{\Omega}, \C^2)$ satisfying
the boundary conditions, 
\begin{align*}
 \|H_\infty\varphi\|^2&:=  \int_\Omega\big(\tfrac{1}{\ri}\pauli\cdot\nabla\varphi(\bx),\tfrac{1}{\ri}\pauli\cdot\nabla\varphi(\bx)\big)
  d\bx=
 \int_\Omega\big(\pauli\cdot\nabla\varphi(\bx),\pauli\cdot\nabla\varphi(\bx)\big) d\bx,
\\
&=\sum_{j\not=k} \int_\Omega\big(\sigma_j\partial_j\varphi(\bx),\sigma_k\partial_k\varphi(\bx)\big)
   d\bx+
\sum_{j=1}^3  \int_\Omega\big(\sigma_j\partial_j\varphi(\bx),\sigma_j\partial_j\varphi(\bx)\big)
   d\bx\\
&=: T_1+T_2.
\end{align*}
For the second term above we have $
 T_2=\|\nabla\varphi\|_\Omega^2
$.  Moreover, using that $\sigma_1\sigma_2=-\sigma_2\sigma_1$ and
Green's identity we obtain 
\begin{align*}
  T_1&=
 \int_\Omega\big(\partial_1\varphi(\bx),\sigma_1\sigma_2\partial_2\varphi(\bx)\big)
   d\bx+
\int_\Omega\big(\partial_2\varphi(\bx),\sigma_2\sigma_1\partial_1\varphi(\bx)\big)
   d\bx\\
&=\int_\Omega\Big[
\partial_1
 \big(\varphi(\bx),\sigma_1\sigma_2\partial_2\varphi(\bx)\big)
+\partial_2
\big(\varphi(\bx),\sigma_2\sigma_1\partial_1\varphi(\bx)\big)
\Big] d\bx\\
&=\int_0^L\Big[
 \big(\varphi(\vg(s)),\sigma_1\sigma_2\partial_2\varphi(\vg(s))\big)\cos\alpha
 (s)+
\big(\varphi(\vg(s)),\sigma_2\sigma_1\partial_1\varphi(\vg(s))\big)\sin{\alpha}(s)
\Big] ds.
\end{align*}
For $s\in [0,L]$ we set
\begin{align*}
  S_1(s):=\big(\varphi(\vg(s)),\sigma_1\sigma_2\partial_2\varphi(\vg(s))\big)\cos\alpha
 (s)+
\big(\varphi(\vg(s)),\sigma_2\sigma_1\partial_1\varphi(\vg(s))\big)\sin{\alpha}(s).
\end{align*}
A simple computation yields (in a slight abuse of notation we write
$\varphi$ for $\varphi(\vg(\cdot))$)
\begin{align*}
  S_1&=(\ri \overline{\varphi}_1\partial_2\varphi_1-\ri
\overline{\varphi}_2\partial_2\varphi_2)\cos\alpha
+(-\ri \overline{\varphi}_1\partial_1\varphi_1+\ri
\overline{\varphi}_2\partial_1\varphi_2)\sin\alpha\\
&=\frac{1}{2}\big[-e^{\ri
    \alpha}\overline{\varphi}_1(\partial_1-\ri\partial_2)\varphi_1+e^{-\ri
    \alpha}\overline{\varphi}_1(\partial_1+\ri\partial_2)\varphi_1\\
&\quad+e^{\ri
    \alpha}\overline{\varphi}_2(\partial_1-\ri\partial_2)\varphi_2-e^{-\ri
    \alpha}\overline{\varphi}_2(\partial_1+\ri\partial_2)\varphi_2\big].
\end{align*}
Using \eqref{eqq:3} and \eqref{eq:4} we see that at the boundary
\begin{align*}
  \partial_1\pm\ri\partial_2=e^{\pm \ri \alpha}(\partial_t\pm
  \ri\partial_s).
\end{align*}
Therefore, 
\begin{align*}
   S_1=&\frac{1}{2}\big[-\overline{\varphi}_1(\partial_t-\ri\partial_s)\varphi_1
+\overline{\varphi}_1(\partial_t+\ri\partial_s)\varphi_1+\overline{\varphi}_2(\partial_t-\ri\partial_s)\varphi_2-\overline{\varphi}_2(\partial_t+\ri\partial_s)\varphi_2\big]\\
=&\ri(\overline{\varphi}_1\partial_s\varphi_1 -\overline{\varphi}_2\partial_s\varphi_2).
\end{align*}
Using the boundary conditions we obtain
\begin{align*}
  \overline{\varphi}_2\partial_s\varphi_2&=-\ri e^{-\ri\alpha}\overline{\varphi}_1
\partial_s
(\ri e^{\ri\alpha}
  \varphi_1)=e^{-\ri\alpha}\overline{\varphi}_1
\partial_s
( e^{\ri\alpha}
  \varphi_1)=\overline{\varphi}_1(\ri \alpha' \varphi_1+\partial_s\varphi_1).
\end{align*}
This implies that
\begin{align*}
  S_1(s)=\alpha'(s) |\varphi_1(\vg(s))|^2,\quad s\in[0,L].
\end{align*}
Thus, we obtain that 
\begin{align}\label{qf11}
   \|H_\infty\varphi\|^2=\int_\Omega |\nabla\varphi(\bx)|^2 d\bx+\int_0^L \alpha'(s) |\varphi_1(\vg(s))|^2  ds.
\end{align}
From this follows \eqref{qf}, since $\varphi\in \mathcal{D}_{\infty}$
and hence $|\varphi_1(\vg(s))|^2=|\varphi(\vg(s))|^2/2$. Thanks to
Lemma \ref{limitarg} the statement
remains true for any $\varphi\in H^1(\Omega,\C^2)$. 
\end{proof}
\section{Lower bound for an auxiliar functional}\label{functional}
\begin{lemma}
\label{variation}
For $\delta>0$ let $f: [0,\delta]\to \R\in H^1$ with $f(\delta)=0$ and $\beta,k\in \R$ with
$|\beta|<1$ and $\delta |\beta|<1/4$. Define
 \begin{align}
  \label{eq:16}
  L[f]:=\int_{0}^\delta (f'(t)^2+k^2 f(t)^2)(1+\beta t) dt .
\end{align}
Then, as $ k\to \infty$, we have 
\begin{align}
  \label{eq:22}
  L[f]\ge f(0)^2 [k +\beta/2]+f(0)^2 \mathcal{O}(e^{-k
  \delta}) +R[f],
\end{align}
where 
\begin{equation}\label{RR}
R[f]=\left\{
\begin{array}{ll}
\frac{k^2}{16}\|f\|^2,&\quad \|f\|^2>\frac{2}{k} f(0)^2\\
 0,&\quad \|f\|^2\le \frac{2}{k} f(0)^2.
\end{array}
\right.
\end{equation}
\end{lemma}
\begin{proof}
We do the substitution $y=kt$ and write $\hat{f}(y)=f(y/k) $ in the
integral in \eqref{eq:16} to get
\begin{equation}
\begin{split}
  \label{eq:24}
    L[f]=\hat{L}[\hat{f}]&:=k \int_{0}^{k\delta}(\hat{f}'(y)^2+\fh (y)^2)dy
+\beta \int_0^{k\delta} y(\fh'(y)^2+\fh^2(y)) dy \\
&=: k L_1[\hat{f}]+\beta L_2[\hat{f}]
\end{split}
\end{equation}
Let $g_0$ be the minimizer of  $L_1$ in $C^2([0,k\delta])$ subject to
the boundary conditions   
$g_0(0)=f(0), g_0(k\delta)=0$. Define   $h=\hat{f}-g_0$  and note that
$h(0)=h(k\delta)=0$. 
Then 
\begin{align*}
  \hat{L}[g_0+h]\ge& \hat{L}[g_0]+\hat{L}[h]+2k \int_0^{k\delta}
(g'_0(y) h'(y)+g_0(y) h(y))dy \\
&\quad +2\beta \int_0^{k\delta}
(g'_0(y) h'(y)+g_0(y) h(y)) ydy.
\end{align*}
Since $g_0$ satisfies the Euler-Lagrange equations integration by
parts yields 
\begin{align*}
  & \int_0^{k\delta}
(g'_0(y) h'(y)+g_0(y) h(y))dy =0,\\
& 2\beta\int_0^{k\delta}
(g'_0(y) h'(y)+g_0(y) h(y)) ydy=-2\beta \int_0^{k\delta} 
g'_0(y)h(y) dy.
\end{align*}
By Schwarz inequality we get
\begin{align*}
 2\beta \int_0^{k\delta} |g'_0(y)|\,|h(y)| dy \le 
\frac{2\beta^2 \| g'_0\|^2}{k}+\frac{k\|h\|^2}{2}.  
\end{align*}
Therefore, 
\begin{equation}
  \label{eq:46}
   \hat{L}[g_0+h]\ge \hat{L}[g_0]-\frac{2\beta^2 \| g'_0\|^2}{k}+\hat{L}[h]-\frac{k\|h\|^2}{2}.
\end{equation}
Since $\delta|\beta|<1/4$ holds $|\beta L_2[h]|\le k/4 L_1[h] $. Using
this together with the fact that $L_1[h]\ge \|h\|^2$ we have that
\begin{align*}
  \hat{L}[h]-\frac{k\|h\|^2}{2}&\ge \frac{3k}{4} L_1[h]-\frac{k\|h\|^2}{2}\ge 
   \frac{k}{4}\|h\|^2.
\end{align*}
In addition, we have 
 \begin{align*}
     \hat{L}[g_0] -\frac{2\beta^2 \| g'_0\|^2}{k}\ge 
 (k-\frac{2\beta^2}{k} ) L_1[g_0] +\beta L_2[g_0].
\end{align*}
Next we use that  $g_0(y)=c_1(k) e^{-y}+c_2(k) e^{y}$, where the
constants (which are determined from the boundary conditions $g_0(0)=f(0)$ 
and $g_0(k\delta)=0$) are given by
\begin{align*}
  c_1(k)=\frac{e^{\delta k}f(0)}{e^{\delta k}-e^{-\delta k}},\quad c_2(k)=\frac{-e^{-\delta k}f(0)}{e^{\delta k}-e^{-\delta k}}.
\end{align*}
We note that, as $k\to\infty$, 
\begin{equation}\label{normg}
\|g_0\|^2=1/2 f(0)^2(1+\mathcal{O}(e^{-k
  \delta})),
\end{equation}
\begin{align*}
   L_1[g_0]&=c_1^2(k)(1-e^{-k\delta})+
c_2^2(k)(1-e^{+k\delta})=f(0)^2
(1 + \mathcal{O}(e^{-k \delta}))\\
 L_2[g_0]&=\tfrac{1}{2} c_1^2(k)(1-e^{-k\delta})-
\tfrac{1}{2} c_2^2(k)(1-e^{+k\delta})= f(0)^2(\tfrac{1}{2} + \mathcal{O}(e^{-k \delta})).
\end{align*}
Therefore, altogether gives
\begin{equation}\label{number}
\begin{split}
  \hat{L}[g_0+h]&\ge \frac{k}{4}\|h\|^2+
  (k-\frac{2\beta^2}{k} ) f(0)^2(1 + \mathcal{O}(e^{-k \delta})) \\
&\quad+\beta f(0)^2(\tfrac{1}{2} + \mathcal{O}(e^{-k \delta}))\\
&=(k+\frac{\beta}{2})f(0)^2+f(0)^2\mathcal{O}(e^{-k \delta})+\frac{k}{4}\|h\|^2.
\end{split}
\end{equation}
Notice that if $\|\hat{f}\|^2> 2 f(0)^2$ then according to
\eqref{normg} $\|h\|^2\ge \|\hat{f}\|^2/4+\mathcal{O}(e^{-k
  \delta})=k \|{f}\|^2/4+\mathcal{O}(e^{-k
  \delta})$ which together with \eqref{number} implies the statement
of the lemma.
\end{proof}
\end{appendix}
%%%%%%%%%%%%%%%%%%%%% 
%%%%%%%%%%%%%%%%%%%%%
%%%%%%%%%%%%%%%%%%%%%

%\bibliographystyle{plain} 
%\bibliography{library}
%\begin{thebibliography}{10}
%\end{thebibliography}
\end{document}